\documentclass[12pt]{article}
\usepackage[utf8]{inputenc}
\usepackage[dvipsnames]{xcolor}
\usepackage{amsthm,amsmath,amssymb,bbm,hyperref,graphicx,float,tikz,comment}
\usepackage[normalem]{ulem} 
\usepackage[affil-it]{authblk}
\allowdisplaybreaks

\newtheorem{thm}{Theorem}
\newtheorem{prop}{Proposition}
\newtheorem{lemma}{Lemma}

\theoremstyle{definition}
\newtheorem{rmk}{Remark}

\theoremstyle{remark}

\DeclareMathOperator{\tr}{tr}
\DeclareMathOperator{\Var}{Var}

\newcommand{\Hilbert}{\mathcal{H}}
\newcommand{\GAP}{\textup{GAP}}

\date{October 1, 2025}

\title{Long-Time Behavior of Typical Pure States from Thermal Equilibrium Ensembles}

\author{Cornelia Vogel\thanks{ORCID: 0000-0002-3905-4730, E-mail: cornelia.vogel@unimi.it}}
\affil{Dipartimento di Matematica, Universit\`a degli Studi di Milano, Via Cesare Saldini 50, 20133 Milano, Italy}

\begin{document}

\maketitle

\begin{abstract}
We consider an isolated macroscopic quantum system in a pure state $\psi_t$ evolving unitarily in a separable Hilbert space $\mathcal{H}$ and take for granted that different macro states $\nu$ correspond to mutually orthogonal subspaces $\mathcal{H}_\nu\subset\mathcal{H}$. Let $P_\nu$ be the projection to $\mathcal{H}_\nu$. It was recently shown that for all Hamiltonians with no highly degenerate eigenvalues and gaps most $\psi_0\in\Hilbert_\mu$  are such that for most $t\geq 0$, $\|P_\nu\psi_t\|^2$ is close to a $t$- and $\psi_0$-independent value $M_{\mu\nu}$ provided that $M_{\mu\nu}$ is not too small. Here, ``most'' refers to the uniform distribution on the sphere $\mathbb{S}(\Hilbert_\mu)$. In the present work, we generalize this result from the uniform distribution, corresponding to the micro-canonical ensemble, to the much more general class of Gaussian adjusted projected (GAP) measures. For any density matrix $\rho$ on $\mathcal{H}$, $\mathrm{GAP}(\rho)$ is the most spread out distribution on $\mathbb{S}(\mathcal{H})$ with density matrix $\rho$. We show that also for $\mathrm{GAP}(\rho)$-most $\psi_0\in\mathcal{H}$ for most $t\geq 0$, $\|P_\nu\psi_t\|^2$ is close to a fixed value $M_{\rho P_\nu}$ (which must not be too small). Moreover, we prove a generalization for certain operators $B$ instead of $P_\nu$ and for finite times. Since certain GAP measures are quantum analogs of the (grand-)canonical ensemble, our result expresses a version of equivalence of ensembles.

\medskip


\end{abstract}

\section{Introduction}
We consider a closed macroscopic quantum system in a pure state and investigate its long-time behavior. This approach in order to study equilibration and thermalization has been very fruitful over the last three decades; important discoveries include \textit{dynamical typicality} \cite{BG09,MGE11,BRGSR18,Reimann2018a,Reimann2018b,RG20,TTV23phys}, \textit{canonical typicality} \cite{Lloyd,GM03,GMM04,PSW06,GLTZ06,TTV24}, the \textit{eigenstate thermalization hypothesis} (ETH) from a physicist's as well as a mathematician's point of view \cite{Deutsch91,Srednicki94,Erdoes21,CEH23,ER24}, and, only very recently, the rigorous study of the thermalization of a system of free fermions \cite{ST24,T24,T24b,RTTV24}. 

Another important result was the rediscovery and further elaboration of \textit{normal typicality} \cite{GLMTZ10,GLTZ10,Reimann2015} which goes back to von Neumann~\cite{vonNeumann29}. The goal of the present work is to generalize (a generalization of) normal typicality from the uniform distribution on the sphere to the so-called \textit{Gaussian adjusted projected (GAP) measures}. To any density matrix $\rho$ on $\Hilbert$ we can associate a measure $\GAP(\rho)$ which is, roughly speaking, the most spread out distribution over $\mathbb{S}(\Hilbert)$ with density matrix $\rho$. For certain density matrices these measures arise as the thermal (and chemical) equilibrium distribution of wave functions and they can be regarded as quantum analogs of the canonical (and grand-canonical) ensemble from classical statistical mechanics \cite{GLTZ06b,GLMTZ15,ITV25}. A mathematically rigorous definition of $\GAP(\rho)$ as well as precise notations can be found in Section~\ref{sec: background}.

\subsection{Normal Typicality}\label{subsec: NT}
We start with recalling the notion of normal typicality and first introduce the setting: 
Following von Neumann~\cite{vonNeumann29}, we take for granted that the system's Hilbert space $\Hilbert$ can be decomposed into mutually orthogonal subspaces $\Hilbert_\nu$ (``macro spaces'') associated with different macro states $\nu$ such that
\begin{align}
    \Hilbert=\bigoplus_\nu \Hilbert_\nu.
\end{align}
Von Neumann thought of the macro spaces $\Hilbert_\nu$ as being the joint eigenspaces of a set of mutually commuting self-adjoint and highly degenerate operators $M_1,\dots,M_K$ (``macroscopic observables''). Each macro space is then characterized by a list of eigenvalues $\nu=(m_1,\dots,m_K)$ of the operators $M_1,\dots,M_K$ and the elements from a macro space $\Hilbert_\nu$ look ``macroscopically the same'' in the sense that every $\psi\in\Hilbert_\nu$ is an eigenfunction of all $M_j$ with corresponding eigenvalues $m_j$. 

Usually, a coarse-grained version of the system's Hamiltonian $H$ is among the $M_j$ and therefore each $\Hilbert_\nu$ is a subset of an \textit{energy shell} $\Hilbert_{\mathrm{mc}}=\mathbbm{1}_{[E-\Delta E, E]}(H)\Hilbert$. Here, $[E-\Delta E, E]$ is a micro-canonical energy interval with $\Delta E$ (the resolution of macroscopic energy measurements) being small on the macroscopic scale but sufficiently large on the microscopic scale such that $[E-\Delta E,E]$ contains a very large (but typically finite) number of eigenvalues\footnote{Note that here and in the following we always assume that $H$ has pure point spectrum.} of $H$ each with finite multiplicity.
This implies in particular that even if we allow $\Hilbert$ to be infinite-dimensional, the $\Hilbert_\nu$ should be of finite dimension. Moreover, in each energy shell there usually is one macro space that contains most dimensions of the shell and which we associate with thermal equilibrium; we denote it by $\Hilbert_{\mathrm{eq}}$.

Let $P_\nu$ be the orthogonal projection to $\Hilbert_\nu$ and let $\mathbb{S}(\Hilbert) = \{\psi\in\Hilbert: \|\psi\|=1\}$ denote the sphere in $\Hilbert$. If $D:=\dim\Hilbert<\infty$, for most $\phi\in\mathbb{S}(\Hilbert)$ (where ``most'' refers to the uniform distribution on $\mathbb{S}(\Hilbert)$) we have that
\begin{align}
    \|P_\nu\phi\|^2 \approx \frac{d_\nu}{D} \quad \forall \nu,\label{eq: NT}
\end{align}
where $d_\nu := \dim\Hilbert_\nu$ provided that $d_\nu$ and $D$ are sufficiently large \cite{GLMTZ10}. If the eigenbasis of $H$ is chosen purely randomly (i.e., according to the Haar measure) among all orthonormal bases (and under some further technical not very restrictive assumptions), every initial wave function $\psi_0\in\mathbb{S}(\Hilbert)$ evolves such that for most times $t\geq 0$,
\begin{align}
    \|P_\nu\psi_t\|^2 \approx \frac{d_\nu}{D},
\end{align}
where $\psi_t = \exp(-iHt)\psi_0$ (we set $\hbar=1$),
provided that $d_\nu$ and $D$ are sufficiently large~\cite{GLMTZ10,GLTZ10,Reimann2015,vonNeumann29}. This phenomenon is known as ``normal typicality''.

In particular, if the macro state $\nu$ represents thermal equilibrium, we write $d_\nu=d_{\mathrm{eq}}$ and have that for every $\psi_0\in\mathbb{S}(\Hilbert)$,
\begin{align}
    \|P_{\mathrm{eq}}\psi_t\|^2 \approx \frac{d_{\mathrm{eq}}}{D} \approx 1
\end{align}
for most $t\geq 0$, i.e., every state spends most of the time in thermal equilibrium. Therefore we have thermalization in the sense that also non-equilibrium initial states sooner or later reach, at least approximately, the thermal equilibrium macro space and stay there for most of the time. This sense of thermal equilibrium is often called ``macroscopic thermal equilibrium'' (MATE), see, e.g., \cite{GHLT15,GHLT16,GLMTZ10b,Tas16,Ueda18,GLTZ20,LPSW09}. Note that here it is important that the statement holds true for \textit{every} initial wave function; if it was only true for ``most'' $\psi_0$, the statement would not tell us much about thermalization because ``most'' states in $\mathbb{S}(\Hilbert)$ are in thermal equilibrium anyway.

\subsection{A Generalization of Normal Typicality for Arbitrary Hamiltonians}

Von Neumann's assumption that the eigenbasis of the Hamiltonian $H$ is Haar-distributed is not physically realistic because in this case it is unrelated to the decomposition of $\Hilbert$ into the macro spaces and therefore the system goes very rapidly from any (possibly very non-equilibrium) macro space almost immediately to the thermal equilibrium one~\cite{GHT13,GHT14,GHT15}. However, if the system starts in a very non-equilibrium state, it should pass through larger and larger (and therefore less far away from thermal equilibrium) macro spaces until it finally reaches $\Hilbert_{\mathrm{eq}}$.

In \cite{TTV23phys}, the notion of normal typicality was therefore generalized in the following way: for a general Hamiltonian $H$, for most $\psi_0\in\mathbb{S}(\Hilbert_\mu)$ (where ``most'' again refers to the uniform distribution over $\mathbb{S}(\Hilbert_\mu)$) with a possibly non-equilibrium macro state $\mu$ for most times $t\geq 0$,
\begin{align}
    \|P_\nu\psi_t\|^2 \approx M_{\mu\nu} \quad \forall \nu\label{eq: GNT}
\end{align}
for suitable values $M_{\mu\nu}$ provided that the eigenvalues and eigenvalue gaps of $H$ are not too highly degenerate (``generalized normal typicality''). We remark that the $M_{\mu\nu}$ depend also an $\mu$ but we still expect that in relevant cases, $M_{\mu\nu} \approx d_\nu/D$. 

While \eqref{eq: NT} holds in the sense that the absolute as well as the relative errors are small, \eqref{eq: GNT} was in \cite{TTV23phys} only proved in the sense that the absolute error is small. In order to show that the relative errors are small as well, we need a lower bound on the quantities $M_{\mu\nu}$. In \cite{TTV23math} such a lower bound was obtained in the case that the Hamiltonian is of the form $H=H_0+V$, where $H_0$ is a deterministic Hermitian matrix and $V$ is a Hermitian random Gaussian perturbation, by making use of results from random matrix theory. We remark that \eqref{eq: GNT} can also be shown for arbitrary bounded operators $B$, i.e., with $\|P_\nu\psi_t\|^2$ replaced by $\langle\psi_t|B|\psi_t\rangle$ and $M_{\mu\nu}$ replaced by a suitable quantity $M_{\mu B}$. Moreover, a finite time result is available, however, the times required for a small error are usually extremely large, see, e.g., \cite{TTV23phys} for a detailed discussion.

\subsection{Brief Discussion of the Main Result}
In the present work we generalize \eqref{eq: GNT} in the sense that we replace the notion of ``most'', which there refers to the uniform distribution over the sphere $\mathbb{S}(\Hilbert_\mu)$ of some macro space $\Hilbert_\mu$, with some other measure on $\mathbb{S}(\Hilbert)$, namely with \textit{GAP measures}. The acronym GAP stands for \textit{Gaussian adjusted projected} measure \cite{JRW94,GLTZ06b} and it refers to one possible way of how these measures can be constructed.\footnote{GAP measures we first introduced by Jozsa, Robb and Wootters \cite{JRW94} in 1994 in an information theoretical context. More precisely, they showed that $\GAP(\rho)$ minimizes the so-called ``accessible information'' under the constraint that the density matrix of the ensemble is given by $\rho$. They named it \textit{Scrooge measure}, referring to Ebenezer Scrooge, the protagonist of Charles Dickens' novella \textit{A Christmas Carol} (1843) who is a very stingy character. The authors chose this name as the GAP measure is ``particularly stingy with its information''.}

As mentioned above, to any density matrix $\rho$ on $\Hilbert$ we can associate a measure $\GAP(\rho)$ on $\mathbb{S}(\Hilbert)$, see Section~\ref{sec: background} for a mathematically precise definition and construction of $\GAP(\rho)$ based on its acronym. Note that if $\rho=P_\mu/d_\mu$ then $\GAP(\rho)$ is just the uniform distribution over $\mathbb{S}(\Hilbert_\mu)$ and we recover the generalized normal typicality result from \cite{TTV23phys} (up to some constants). We remark already here that $\GAP(\rho)$ can, in contrast to the uniform distribution, also be defined on separable Hilbert spaces, i.e., the Hilbert space does not have to be of finite dimension. If $\rho=\rho_{\mathrm{can}}$, i.e., if $\rho$ is of the form
\begin{align}
    \rho=\rho_{\mathrm{can}} = \frac{1}{Z} e^{-\beta H},
\end{align}
where $\beta$ is the inverse temperature and $Z$ a normalization constant, $\GAP(\rho)$ arises as the thermal equilibrium distribution of the wave function. More precisely, if $\Hilbert$ is a micro-canonical energy shell, then for most $\psi\in\mathbb{S}(\Hilbert)$ the wave function of a small subsystem\footnote{By ``wave function of a subsystem'' we mean the \textit{conditional wave function} \cite{GLTZ06b, GLMTZ15}.} (that is only weakly interacting with its environment) is typically approximately $\GAP(\rho)$-distributed, where $\rho$ is a canonical density matrix for the subsystem, see \cite{GLTZ06b, GLMTZ15} for details. This justifies to view certain GAP measures as quantum analogs of the canonical ensemble of classical mechanics. We remark that similar considerations can also be made in the grand-canonical setting, see \cite{ITV25}.

Our main result in the present paper is that also for most $\psi_0\in\mathbb{S}(\Hilbert)$, where ``most'' now refers to $\GAP(\rho)$, for most $t\geq 0$,
\begin{align}
    \langle\psi_t|B|\psi_t\rangle \approx M_{\rho B}\label{eq: new result}
\end{align}
with suitable quantities $M_{\rho B}$ provided that $\|\rho\|$ is small, $\|B\|$ is not too large compared to $|M_{\rho B}|$ and the eigenvalues and eigenvalue gaps of the self-adjoint operator $H$, which is assumed to have pure point spectrum, are not too highly degenerate. Here, $B$ is any operator if $\Hilbert$ is finite-dimensional. If $\dim\Hilbert=\infty$, we restrict the class of operators $B$, roughly speaking, to the bounded ones which act non-trivially only on finitely many eigenspaces of $H$ and whose image is contained in the span of only finitely many eigenspaces. 

As discussed above, the projections $P_\nu$ to the macro spaces should fulfill this assumption as the macro spaces are usually contained in an energy shell which typically is finite-dimensional. Note that here we restrict our considerations to the absolute errors; however, if $B=P_\nu$ or $\rho=P_\nu/d_\nu$ for some macro state $\nu$, the relative bounds from \cite{TTV23math} apply. 

Our result shows that for $\GAP(\rho)$-most $\psi_0\in\mathbb{S}(\Hilbert)$ the curve $t\mapsto\langle\psi_t|B|\psi_t\rangle$ is nearly constant in the long run $t\to\infty$. We remark that for $\GAP(\rho)$-most $\psi_0\in\mathbb{S}(\Hilbert)$, the curve $t\mapsto \langle\psi_t|B|\psi_t\rangle$ is nearly deterministic (``dynamical typicality''); this has been shown in \cite{TTV24}. The generalization of (generalized) normal typicality to $\GAP(\rho)$ reveals that it is not a phenomenon specific for the uniform distribution but it is also valid for a broader class of natural and physically relevant distributions which can also be defined on infinite-dimensional Hilbert spaces. As some GAP measures are quantum analogs of the canonical and grand-canonical ensemble whereas the uniform distribution on an energy shell $\Hilbert_{\mathrm{mc}}$ is an analog of the micro-canonical ensemble, our result shows that the phenomenon of generalized normal typicality occurs for all three statistical ensembles. Therefore it can be regarded as expressing a version of equivalence of ensembles.

The main ingredients of our proof are a further improvement of bounds on the $\GAP(\rho)$-variance of $\langle\psi|B|\psi\rangle$ (the first version was obtained by Reimann~\cite{Reimann08} and then slightly improved in \cite{TTV24}) as well as Lévy's Lemma for GAP measures \cite{TTV24}, a concentration-of-measure-type result for Lipschitz continuous functions on the sphere.

\subsection{Structure of the Paper}
The remainder of this paper is organized as follows: In Section~\ref{sec: background}, we give some background and present the mathematical setup. In Section~\ref{sec: main}, we formulate and discuss our main result. In Section~\ref{sec: proof}, we provide the proofs. Finally, in Section~\ref{sec: conclusion}, we conclude.

\section{Background and Mathematical Setup\label{sec: background}}
In the following, we make precise some notions that appeared in the introduction, give a mathematically rigorous definition of the GAP measures and introduce the system's Hamiltonian and its relevant quantities that we will need to state our main result in Section~\ref{sec: main}. Throughout this paper we assume that $\Hilbert$ is a separable Hilbert space, i.e., that it has either a finite or countably infinite orthonormal basis.

\paragraph{Measures of ``most''.}
Let $\mathbb{P}$ be a probability measure on $\mathbb{S}(\Hilbert)$ and let $\varepsilon>0$. We say that, w.r.t. $\mathbb{P}$, a statement $s(\psi)$ is true for $(1-\varepsilon)$-most $\psi\in\mathbb{S}(\Hilbert)$ if
\begin{align}
    \mathbb{P}\left\{\psi\in\mathbb{S}(\Hilbert): s(\psi)\mbox{ holds}\right\} \geq 1-\varepsilon.
\end{align}
Let $T,\delta>0$. Analogously we say that a statement $S(t)$ is true for $(1-\delta)$-most $t\in[0,T]$ if
\begin{align}
    \frac{1}{T}\,\lambda\left\{t\in [0,T]: S(t)\mbox{ holds}\right\} \geq 1-\delta,
\end{align}
where $\lambda$ denotes the Lebesgue measure on $\mathbb{R}$. Moreover, we say that a statement $S(t)$ is true for $(1-\delta)$-most $t\in[0,\infty)$ (or also $t\geq 0$) if
\begin{align}
    \liminf_{T\to\infty} \frac{1}{T}\,\lambda\left\{t\in[0,T]: S(t)\mbox{ holds}\right\} \geq 1-\delta.
\end{align}

\paragraph{Time averages.}
Let $T>0$ and let $f:\mathbb{R}_+\to\mathbb{C}$. We define the finite time average $\langle\,\cdot\,\rangle_T$ of $f$ over the interval $[0,T]$ by
\begin{align}
    \langle f(t)\rangle_T := \frac{1}{T}\int_0^T f(t)\, dt.
\end{align}
Moreover, the infinite time average of $f$ is given by
\begin{align}
    \overline{f(t)} := \lim_{T\to\infty} \frac{1}{T}\int_0^T f(t)\, dt
\end{align}
whenever this limit exists.

\paragraph{Norms.} 
Throughout this paper we need two norms, the operator norm and the trace norm. Let $M$ be an operator on $\Hilbert$. Its \textit{operator norm} is given by
\begin{align}
    \|M\| := \sup_{\|\psi\|=1} \|M\psi\|.
\end{align}
If $M$ is self-adjoint, $\|M\|$ is equal to the largest absolute eigenvalue of $M$. In general, $\|M\|$ is equal to the square root of the largest eigenvalue of $M^*M$ where $M^*$ is the adjoint operator of $M$.

The \textit{trace norm} of $M$ is defined as
\begin{align}
    \|M\|_{\tr} := \tr |M| = \tr\sqrt{M^* M}.
\end{align}
If $M$ is self-adjoint, $\|M\|_{\tr}$ is equal to the sum of the absolute eigenvalues of $M$.

\paragraph{Density matrix.}
Let $\mathbb{P}$ be a probability measure on $\mathbb{S}(\Hilbert)$. From the measure $\mathbb{P}$ we can obtain a density matrix $\rho_{\mathbb{P}}$ on $\Hilbert$ by
\begin{align}
    \rho_{\mathbb{P}} := \int_{\mathbb{S}(\Hilbert)} \mathbb{P}(d\psi)\, |\psi\rangle\langle\psi| 
\end{align}
and we say that the measure $\mathbb{P}$ has density matrix $\rho_{\mathbb{P}}$.
Note that this density matrix always exists, see, e.g., Lemma~1 in \cite{Tum20}. If the mean of $\mathbb{P}$ is equal to zero, $\rho_{\mathbb{P}}$ is the covariance matrix of $\mathbb{P}$.

\paragraph{GAP measure.}
In this paragraph, we give a mathematically rigorous definition of the measure $\GAP(\rho)$ on $\mathbb{S}(\Hilbert)$ for a density matrix $\rho$ on $\Hilbert$. There are several equivalent definitions of $\GAP(\rho)$ in the case that $\Hilbert$ is finite-dimensional, see \cite{GLMTZ15}. In the following, we present the one which gave the GAP measures their name and, for simplicity, we restrict ourselves to the finite-dimensional case. A similar construction can also be done in the case that $\dim\Hilbert=\infty$; we refer to \cite{Tum20} for details. 

Let $\rho$ be a density matrix on a finite-dimensional Hilbert space $\Hilbert$. Then we can write it as
\begin{align}
    \rho = \sum_n p_n |n\rangle\langle n|,
\end{align}
where the $|n\rangle$ form an orthonormal eigenbasis of $\rho$ with corresponding eigenvalues $p_n$.

As the name \textit{Gaussian adjusted projected measure} suggests, we start from a Gaussian measure $G(\rho)$ on $\Hilbert$; it is constructed as follows: Let $(Z_n)$ be a sequence of independent complex-valued centered Gaussian random variables\footnote{Recall that a complex-valued centered Gaussian random variable $Z$ with variance $\sigma^2>0$ has independent real and imaginary part $\mbox{Re}\, Z$ and $\mbox{Im}\, Z$, $\mathbb{E} \mbox{Re}\, Z = \mathbb{E} \mbox{Im}\, Z=0$ and $\mathbb{E}(\mbox{Re} Z)^2 = \mathbb{E}(\mbox{Im}\,Z)^2=\sigma^2/2$.} with variances
\begin{align}
    \mathbb{E}|Z_n|^2 = p_n.
\end{align}
The Gaussian measure $G(\rho)$ is then defined as the distribution of the random vector
\begin{align}
    \Psi^{G} := \sum_n Z_n |n\rangle.
\end{align}
While $G(\rho)$ is centered and has density matrix $\rho$, it is not a distribution on the sphere $\mathbb{S}(\Hilbert)$: It follows immediately from the definition of $\Psi^G$ together with the fact that the $p_n$ sum up to 1 that $\mathbb{E}\|\Psi^G\|^2=1$ but in general $\|\Psi^G\|\neq 1$. 

In the next step, we have to adjust the measure $G(\rho)$; this adjustment is necessary because if we would project $G(\rho)$ directly to the sphere $\mathbb{S}(\Hilbert)$, the resulting measure would not have the property that it has density matrix $\rho$. We define the \textit{adjusted Gaussian measure} $\mbox{GA}(\rho)$ on $\Hilbert$ by
\begin{align}
    \mbox{GA}(\rho)(d\psi) := \|\psi\|^2 G(\rho)(d\psi),
\end{align}
i.e., we multiply the density of $G(\rho)$ by the factor $\|\psi\|^2$. It follows from $\mathbb{E}\|\Psi^G\|^2=1$ that also $\mbox{GA}(\rho)$ is a probability measure on $\Hilbert$.

In the last step of the construction, we project the measure $\mbox{GA}(\rho)$ on $\Hilbert$ to the sphere $\mathbb{S}(\Hilbert)$. To this end, let $\Psi^{\mathrm{GA}}$ be a $\mbox{GA}(\rho)$-distributed random vector. Then we define $\GAP(\rho)$ to be the distribution of the random vector
\begin{align}
    \Psi^{\GAP} := \frac{\Psi^{\mathrm{GA}}}{\|\Psi^{\mathrm{GA}}\|}.
\end{align}
Obviously, this measure is a distribution on $\mathbb{S}(\Hilbert)$ and a short computation shows that $\GAP(\rho)$ indeed has density matrix $\rho$:
\begin{align}
    \rho_{\GAP(\rho)} = \int_{\mathbb{S}(\Hilbert)} \GAP(\rho)(d\psi) |\psi\rangle\langle\psi| = \int_{\Hilbert} \mbox{GA}(\rho)(d\psi) \frac{1}{\|\psi\|^2} |\psi\rangle \langle\psi| = \int_{\Hilbert} G(\rho)(d\psi) |\psi\rangle\langle\psi| =\rho. 
\end{align}

\paragraph{The Hamiltonian.}
Let $H$ be a self-adjoint operator on $\Hilbert$ with pure point spectrum.
We write $H$ in the spectral decomposition
\begin{align}
    H=\sum_{e\in\mathcal{E}} e\, \Pi_e,
\end{align}
where $\mathcal{E}$ denotes the set of distinct eigenvalues of $H$ and $\Pi_e$ is the projection onto the eigenspace of $H$ with corresponding eigenvalue $e$. We define $d_E:=|\mathcal{E}|$, where $|\{\cdot\}|$ denotes the number of elements of the set $\{\cdot\}$, and $D_E:=\max_{e\in\mathcal{E}}\tr(\Pi_e)$, i.e., $d_E$ is the number of distinct eigenvalues of $H$ and $D_E$ is the maximum degeneracy of an eigenvalue of $H$. Moreover, we define the maximal gap degeneracy $D_G$ by
\begin{align}
    D_G := \max_{E\in\mathbb{R}}\left|\left\{(e,e')\in\mathcal{E}\times\mathcal{E}: e\neq e' \mbox{ and } e-e'=E\right\} \right|.
\end{align}
Let $\kappa>0$. The maximal number of gaps in an energy interval of length $\kappa$ is given by
\begin{align}
    G(\kappa) := \max_{E\in\mathbb{R}}\left|\left\{(e,e')\in\mathcal{E}\times\mathcal{E}: e\neq e' \mbox{ and } e-e' \in [E,E+\kappa) \right\}\right|
\end{align}
and obviously $D_G = \lim_{\kappa\to 0^+} G(\kappa)$.

If $\Hilbert$ is infinite-dimensional, these quantities are not necessarily finite. In particular, $d_E$ and $D_E$ cannot both be finite. Let $B$ be a bounded operator on $\Hilbert$. We will see in the proofs that only those eigenvalues $e,e'$ with $\Pi_e B \Pi_{e'}\neq 0$ contribute to the quantities we want to compute and we will require that there are only finitely many such $e,e'$ which will imply that the sums over $e,e'$ are effectively finite. This is equivalent to assuming that there are only finitely many eigenvalues $e$ such that $\Pi_e B \neq 0$ or $B\Pi_e \neq 0$. We therefore define for an operator $B$ the set of ``contributing'' eigenvalues by
\begin{align}
    \mathcal{E}_B &:= \{e\in\mathcal{E}: \Pi_e B \neq 0 \mbox{ or } B\Pi_e \neq 0\}.
\end{align}
Note that if all eigenvalues of $H$ have finite multiplicity (which usually is the case, see also Section~\ref{subsec: NT}), $|\mathcal{E}_B|<\infty$ implies that $B$ is of finite rank (and therefore bounded).

 We also define the quantities $d_E, D_E, D_G$ and $G(\kappa)$ relative to $B$:
\begin{align}
    d_{E,B} &:= |\mathcal{E}_B|,\\
    D_{E,B} &:= \max_{e\in\mathcal{E}_B} \tr(\Pi_e),\\
    D_{G,B} &:= \max_{E\in\mathbb{R}}\left|\left\{(e,e')\in\mathcal{E}_B\times\mathcal{E}_B: e\neq e' \mbox{ and } e-e'=E \right\}\right|,\\
    G_B(\kappa) &:= \max_{E\in\mathbb{R}}\left|\left\{ (e,e')\in\mathcal{E}_B\times\mathcal{E}_B: e\neq e' \mbox{ and } e-e'\in[E,E+\kappa)\right\}\right|.
\end{align}

\section{Main Result \label{sec: main}}
\subsection{Statement}
After having introduced the relevant quantities and notations in the previous section, we are now ready to state our main result.

\begin{thm}[Normal Typicality for $\GAP(\rho)$]\label{thm: NT GAP}
Let $\Hilbert$ be a separable Hilbert space of dimension $\geq 4$.
    Let $B$ be a bounded operator with $d_{E,B}<\infty$ and let $\rho$ be a density matrix with $\|\rho\|<1/4$. Let $\varepsilon,\delta,\kappa,T>0$ and define
    \begin{align}
        M_{\rho B} := \sum_{e\in\mathcal{E}} \tr(\rho \Pi_e B \Pi_e).
    \end{align}
    Then, w.r.t. $\GAP(\rho)$, $(1-\varepsilon)$-most $\psi_0\in\mathbb{S}(\Hilbert)$ are such that for $(1-\delta)$-most $t\in [0,T]$
    \begin{align}
        \Bigl|\langle\psi_t|B|\psi_t\rangle - M_{\rho B}\Bigr| &\leq \left(\min\left\{\frac{188}{\varepsilon\delta},\frac{25\log(24/(\varepsilon\delta))}{\delta C}\right\} \|B\|^2 \|\rho\| D_{E,B} G_B(\kappa)\left(1+\frac{8\log_2 d_{E,B}}{\kappa T}\right)\right)^{1/2},\label{ineq: MR finite time}
    \end{align}
   where $C=\frac{1}{288\pi^2}$.
Moreover, w.r.t. $\GAP(\rho)$, $(1-\varepsilon)$-most $\psi_0\in\mathbb{S}(\Hilbert)$ are such that for $(1-\delta)$-most $t\in[0,\infty)$,
\begin{align}
    \Bigl|\langle\psi_t|B|\psi_t\rangle - M_{\rho B}\Bigr| &\leq \left(\frac{188}{\varepsilon\delta} \|B\|^2 \|\rho\| D_{E,B} D_{G,B}\right)^{1/2}.\label{ineq: MR infinite time}
\end{align}
\end{thm}

\subsection{Discussion}

\begin{rmk}{\it Smallness of the bounds.}
The right-hand side of \eqref{ineq: MR infinite time} is small compared to $\|B\|$ as soon as $\|\rho\| \ll \left(D_{E,B} D_{G,B}\right)^{-1}$, i.e., as soon as the largest eigenvalue of $\rho$ is small and no eigenvalue and no gap of $H$ is hugely degenerate. If additionally $\|B\| \not\gg |M_{\rho B}|$, i.e., if $|M_{\rho\frac{B}{\|B\|}}|$ is not too small, then also the relative error is small and we can conclude that for most initial states $\psi_0\in\mathbb{S}(\Hilbert)$, where ``most'' refers to $\GAP(\rho)$, the expectation value $\langle\psi_t|B|\psi_t\rangle$ of $B$ in the state $\psi_t$ is close to the fixed value $M_{\rho B}$ for most times $t\geq 0$. Put another way, the curve $t\mapsto \langle\psi_t|B|\psi_t\rangle$ is nearly constant in the long run $t\to\infty$. Note that if $B=P_\nu$ or $\rho=P_\nu/d_\nu$ for some macro state $\nu$, sufficient conditions on $H$ that ensure that $|M_{\rho \frac{B}{\|B\|}}|$ is not too small can be found in \cite{TTV23math}.

The finite-time result \eqref{ineq: MR finite time} shows that we additionally need that the time $T$ is large enough in order to obtain small errors. Unfortunately, the times $T$ required to ensure that the term $G_B(\kappa)(1+8\log_2 d_{E,B}/(\kappa T))$ is not too large are usually huge: for example for a system of $N$ particles we need that $T\gg \exp(N)$, see, e.g., the discussion in \cite{TTV23phys} after Theorem~4.

In the case of finite times, we showed two bounds which depend differently on $\varepsilon$ and $\delta$. While the first one follows from proving upper bounds for suitable expectations and variances and then applying Markov's and Chebyshev's inequality, the second one is a consequence of Lévy's Lemma for GAP measures~\cite{TTV24}. If $\varepsilon \ll \delta$, then the second bound gives a better result while otherwise the first bound is smaller.
    \hfill$\diamond$
\end{rmk}

\begin{rmk}{\it Strategy of proof in the infinite-time case.}
In the case of infinite times we cannot make use of the strategy of proof involving Lévy's Lemma. The reason is that in one step we have to interchange the quantifiers ``most $t$'' and ``most $\psi_0$'', which in the finite time case is possible by Footnote~7 in \cite{GLMTZ10} (at the expense of making the sets of admissible $t$ and $\psi_0$ smaller) as both ``most'' refer to probability measures. However, when considering the infinite time interval $[0,\infty)$, we are dealing with an infinite measure space instead of a finite one and ``most $t$'' does not refer to a probability measure. In this situation we can in general not interchange ``most $t$'' and ``most $\psi_0$''.\footnote{\label{footnote counterex}While a statement that is true for most $\psi_0\in\mathbb{S}(\Hilbert)$ for most $t\in[0,\infty)$ is also true for most $t\in[0,\infty)$ for most $\psi_0\in\mathbb{S}(\Hilbert)$ (which can be proved similarly as Footnote~7 in \cite{GLMTZ10} and by making use of Fatou's Lemma), the converse does not hold. In general, the quantifier ``most $t\in[0,\infty)$''  and a ``most''-quantifier referring to a probability measure on a probability space cannot be interchanged. As a simple example consider the interval $[0,1]$ equipped with the Lebesgue measure and let $N\in\mathbb{N}$ be large. We define $I_1 := [0,1/N)^c, I_2 := [1/N,2/N)^c, \dots, I_N := [(N-1)/N,1]^c$, where the complement is taken in $[0,1]$. We define the set
\begin{align*}
    S:=\bigcup_{m=0}^{\infty} \bigcup_{n=1}^N\left( I_n \times \Bigl[2^{mN+n-1},2^{mN+n}\Bigr)\right)\subset [0,1] \times [0,\infty).
\end{align*}
Clearly, for every $t>1$ we have that $(\psi,t) \in S$ for $(1-1/N)$-most $\psi\in[0,1]$. Therefore, the statement ``$(\psi,t)\in S$'' is true for most $t\in[0,\infty)$ for most $\psi\in [0,1]$.

Let $\psi\in [0,1]$. Without loss of generality we assume that $\psi\in [0,1/N)$ which implies that $\psi\notin I_1$ and $\psi\in I_k$ for all $k\geq 2$. We define a sequence of times $(T_l)_{l\geq 0}$ by $T_l := 2^{lN+1}$. Then, by construction, $(\psi,t)\notin S$ for all $t\in [2^{lN},2^{lN+1})$ and thus $(\psi,t)\notin S$ for at least half of the times $t\in[0,T_l]$. This implies
\begin{align*}
    \limsup_{T\to\infty} \frac{1}{T}\lambda\left\{t\in[0,T]: (\psi,t)\notin S\right\} \geq \lim_{l\to\infty} \frac{1}{T_l}\lambda\left\{t\in [0,T_l]: (\psi,t)\notin S\right\} \geq 1/2.
\end{align*}
As $\psi\in[0,1]$ was arbitrary, we conclude that the statement ``$(\psi,t)\in S$'' is false for every $\psi\in [0,1]$ for a substantial amount of times $t\in[0,\infty)$.}
     \hfill$\diamond$
\end{rmk}

\begin{rmk}{\it Recovering old results.}
If $\rho=P_\mu/d_\mu$ for some macro state $\mu$, the measure $\GAP(\rho)$ becomes the uniform distribution over the sphere $\mathbb{S}(\Hilbert_\mu)$, we have $\|\rho\|=1/d_\mu$ and in this case we basically recover the results from \cite{TTV23phys}. Note that in Theorem~\ref{thm: NT GAP} for any $B$ only the eigenvalues in $\mathcal{E}_B$ and gaps in $\mathcal{E}_B\times\mathcal{E}_B$ contribute to the quantities $D_{E,B}, d_{E,B}, D_{G,B}$ and $G_{B}(\kappa)$. 
However, if $\Hilbert$ is finite-dimensional, these quantities can obviously be replaced by $D_E, d_E, D_G$ and $G(\kappa)$.
    \hfill$\diamond$
\end{rmk}

\begin{rmk}{\it Assumption on $B$.}
The technical assumption on $B$ that $|\mathcal{E}_B|<\infty$ is convenient as it ensures that the sums over $e$ appearing in the proofs are effectively finite and it therefore is for example unproblematic to interchange the sums over $e$ with the time average. Moreover, it covers the important case that $B=P_\nu$ for some macro state $\nu$, see Remark~\ref{rmk: BPnu}. While it might be possible to relax the assumption on $B$, we believe that it cannot be dropped completely. E.g. the computations in the beginning of the proof of Proposition~\ref{prop: Var inequalities} show that the equilibration time is closely connected to the structure of the gaps of $H$. More precisely, it suggests that the more gaps are very close to each other, the larger the equilibration time gets. Therefore it seems plausible that, e.g., $G(\kappa)=\infty$ for all sufficiently small $\kappa>0$ could potentially destroy any equilibration.
    \hfill$\diamond$
\end{rmk}

\begin{rmk}{\it The case $B=P_\nu$.}\label{rmk: BPnu}
Motivated by the phenomenon of normal typicality \cite{GLMTZ10,GLTZ10,Reimann2015,vonNeumann29}, in \cite{TTV23phys} the case that $B=P_\nu$ for some macro state $\nu$ was of particular interest. As discussed in the introduction, the $P_\nu$ project to macro spaces $\Hilbert_\nu$ which usually are contained in a finite-dimensional energy shell. Thus the $P_\nu$ fulfill the assumption of Theorem~\ref{thm: NT GAP} that $|\mathcal{E}_{P_\nu}|<\infty$. We therefore see that not only with respect to the uniform distribution over the sphere of some macro space $\Hilbert_\mu$ but also with respect to $\GAP(\rho)$, for most initial wave functions $\psi_0$ the superposition weights $\|P_\nu\psi_t\|^2$ are, in the long run, close to some fixed value $M_{\rho P_\nu}$ (provided that $M_{\rho P_\nu}$ is not too small). In the case of normal typicality, this value was given by $d_\nu/D$, where $D=\dim\Hilbert$. Here, if $D<\infty$, we also expect that $M_{\rho P_\nu}$ is usually close to $d_\nu/D$; we can argue that this should be the case for example if $H$ is non-degenerate and satisfies the eigenstate thermalization hypothesis in the sense that for every eigenstate $|m\rangle$ of $H$ we have that $\langle m|P_\nu|m\rangle \approx \langle P_\nu\rangle_{\mathrm{mic}}=\tr(P_\nu)/D = d_\nu/D$ where $\langle\,\cdot\,\rangle_{\mathrm{mic}}$ denotes the micro-canonical expectation. Then we find
\begin{align}
    M_{\rho P_\nu} = \sum_{m}\langle m|\rho|m\rangle \langle m|P_\nu|m\rangle \approx \frac{d_\nu}{D} \sum_m \langle m|\rho| m\rangle = \frac{d_\nu}{D}.
\end{align}
If $\dim\Hilbert=\infty$ (and there are infinitely many macro states $\nu$), the $\|P_\nu\psi_t\|^2$ form a null sequence for every $\psi_0$ and $t$. 
\hfill$\diamond$
\end{rmk}

\begin{rmk}{\it Applying Lévy's Lemma.}
Instead of applying Lévy's Lemma to the function $f(\psi_0)=\langle\psi_t|B|\psi_t\rangle$ as in the proof of Theorem~\ref{thm: NT GAP} one might think of applying it directly to $g(\psi_0) = \langle|\langle\psi_t|B|\psi_t\rangle - M_{\rho B}|^2\rangle_T$. In fact, one can show that $g$ is Lipschitz continuous with Lipschitz constant $\eta$ bounded by $8\|B\|^2$ (and the same holds true for the infinite time average as well). However, as it is also explained in \cite{TTV23phys}, in such a situation Lévy's Lemma does in general not lead to better results than applying Markov's and Chebyshev's inequality. Markov's inequality gives
\begin{align}
    g(\psi_0) \leq \frac{\mathbb{E}_\rho g}{\varepsilon}
\end{align}
for $(1-\varepsilon)$-most $\psi_0\in\mathbb{S}(\Hilbert)$ while it follows from Theorem~\ref{thm: Levy Lemma} that 
\begin{align}
    g(\psi_0) \leq \mathbb{E}_\rho g + \sqrt{\frac{2\eta^2 \|\rho\| \log(12/\varepsilon)}{C}}\label{ineq: g LL}
\end{align}
for $(1-\varepsilon)$-most $\psi_0\in\mathbb{S}(\Hilbert)$.

The expectation $\mathbb{E}_\rho g$ is of order $\|\rho\|$ while the bound in \eqref{ineq: g LL} is of order $\|\rho\|^{1/2}$, so while the $\varepsilon$-dependence in the result from Lévy's Lemma is better, the dependence on $\|\rho\|$ is better for Markov's inequality. Often, $\|\rho\|$ is extremely small, for example for $\rho=\mathbbm{1}_\Hilbert/D$ where $D$ is the dimension of the huge (but finite-dimensional) Hilbert space $\Hilbert$ and $\mathbbm{1}_\Hilbert$ the identity on it, we have that $\|\rho\|=1/D$. In such cases, the better $\varepsilon$-dependence does not compensate the worse $\|\rho\|$-dependence unless $\varepsilon$ is extremely small and it is of little interest to consider, e.g., $\varepsilon$ smaller than $10^{-200}$ \cite{Bor62}.
    \hfill$\diamond$
\end{rmk}

\subsection{Outline of the Proof}
We now give a brief outline of the proof of Theorem~\ref{thm: NT GAP}.
Besides Lévy's Lemma for GAP measures which was proved in \cite{TTV24}, an important ingredient for the proof is an upper bound for the variance of $\langle\psi|A|\psi\rangle$ where $A$ is a bounded operator on $\Hilbert$ and $\psi\sim \GAP(\rho)$.
Reimann \cite{Reimann08} proved such an upper bound under the assumption that $A$ is self-adjoint and $\Hilbert$ is finite-dimensional. In \cite{TTV24} the proof was generalized to arbitrary bounded operators $A$ on a separable Hilbert space $\Hilbert$. Unfortunately, the bounds obtained there are not sufficiently sharp for our purpose as they would lead in Theorem~\ref{thm: NT GAP} to an upper bound of the order
\begin{align}
\sqrt{\tr\rho^2 \sum_{e\neq e'}\|\Pi_e B \Pi_{e'}\|^2} \leq  \sqrt{\|\rho\| \sum_{e,e'} \tr(\Pi_{e'} B^* \Pi_e B \Pi_{e'})}=\sqrt{\|\rho\| \tr(B^*B)}
\end{align}
and for example in the special case that $\rho=P_\mu/d_\mu$ and $B=P_\nu$ this would yield an upper bound of the order $\sqrt{d_\nu/d_\mu}$ which is not necessarily small (here we ignored the possible degeneracy of the eigenvalues and gaps of $H$ to simplify the discussion). 

Therefore, using similar techniques as in \cite{Reimann08,TTV24}, we prove in Lemma~\ref{lem: var GAP} a slightly improved upper bound for the variance that allows us to ultimately obtain a bound of the order $\|B\| \|\rho\|^{1/2}$ in Theorem~\ref{thm: NT GAP}. 
With the help of this improved upper bound for the variance of $\langle\psi|A|\psi\rangle$, we compute upper bounds for the expectation w.r.t. $\GAP(\rho)$ of the infinite-time variance of $\langle\psi_t|B|\psi_t\rangle$, a related quantity where the outer infinite-time average is replaced by a finite one, the finite-time average of the squared difference between $\mathbb{E}_\rho\langle\psi_t|B|\psi_t\rangle$ and $M_{\rho B}$ and the variance w.r.t. $\GAP(\rho)$ of the infinite-time average $\overline{\langle\psi_t|B|\psi_t\rangle}$ in Proposition~\ref{prop: Var inequalities}. In these computations, many traces and products and sums over traces appear and we use similar methods as in \cite{TTV23phys} to bound them. By making use of these bounds from Proposition~\ref{prop: Var inequalities} and Markov's and Chebyshev's inequality we are able to prove Theorem~\ref{thm: NT GAP}.

\section{Proofs\label{sec: proof}}
This section is devoted to the proof of our main result. In Section~\ref{subsec: proof main} we give the proof of Theorem~\ref{thm: NT GAP} based on Lévy's Lemma for GAP measures from~\cite{TTV24}, which we quote in Theorem~\ref{thm: Levy Lemma}, Proposition~\ref{prop: Var inequalities} and a careful application of Markov's and Chebyshev's inequality. Proposition~\ref{prop: Var inequalities} provides us with an upper bound for the expected time variance of $\langle\psi_t|B|\psi_t\rangle$ (and for a modification thereof in the finite-time case) as well as for the GAP$(\rho)$-variance of $\overline{\langle\psi_t|B|\psi_t\rangle}$ and the finite-time average of the squared difference between $\mathbb{E}_\rho\langle\psi_t|B|\psi_t\rangle$ and $M_{\rho B}$; we give its proof in Section~\ref{subsec: proof prop 1}. It makes use of an improved upper bound for the $\GAP(\rho)$-variance $\Var_\rho\langle\psi|A|\psi\rangle$ which we state in Lemma~\ref{lem: var GAP} and prove in Section~\ref{subsec: proof Lemma 1}.

\subsection{Proof of Theorem~\ref{thm: NT GAP}\label{subsec: proof main}}
The proof of Theorem~\ref{thm: NT GAP} is based on the following theorem and proposition:

\begin{thm}[Lévy's Lemma for GAP measures \cite{TTV24}]\label{thm: Levy Lemma}
    Let $\Hilbert$ be a separable Hilbert space, let $f:\mathbb{S}(\Hilbert) \to \mathbb{C}$ be a Lipschitz continuous function with Lipschitz constant $\eta$, let $\rho$ be a density matrix on $\Hilbert$ and let $\varepsilon\geq 0$. Then,
    \begin{align}
        \GAP(\rho)\Bigl\{\psi\in\mathbb{S}(\Hilbert): \bigl|f(\psi)-\mathbb{E}_\rho(f) \bigr|>\varepsilon\Bigr\} \leq 12 \exp\left(-\frac{C \varepsilon^2}{2\eta^2 \|\rho\|}\right),
    \end{align}
    where $C=\frac{1}{288\pi^2}$.
\end{thm}

\begin{prop}\label{prop: Var inequalities}
    Let $\Hilbert$ be a separable Hilbert space of dimension $\geq 4$. Let $\rho$ be a density matrix on $\Hilbert$ with eigenvalues $p_n>0$ and $\|\rho\|<1/4$ and let $B$ be a bounded operator on $\Hilbert$ such that $d_{E,B}<\infty$. Then for every $\kappa, T>0$,
    \begin{align}
        \mathbb{E}_\rho\left(\left\langle\left|\langle\psi_t|B|\psi_t\rangle - \overline{\langle\psi_t|B|\psi_t\rangle} \right|^2\right\rangle_T\right)
    &\leq 24\|B\|^2 \|\rho\| D_{E,B} G_{B}(\kappa)\left(1+\frac{8\log_2 d_{E,B}}{\kappa T}\right),\label{ineq: Erho Bpsit}\\
    \left\langle\left|\tr(e^{iHt} B e^{-iHt}\rho)-M_{\rho B} \right|^2\right\rangle_T &\leq \|B\|^2 \|\rho\| D_{E,B} G_B(\kappa) \left(1+\frac{8\log_2 d_{E,B}}{\kappa T}\right),\label{ineq: prop1 second}\\
    \Var_\rho\overline{\langle\psi_t|B|\psi_t\rangle} &\leq 23 \|B\|^2 \|\rho\|.
    \end{align}
 Moreover, 
    \begin{align}
        \mathbb{E}_\rho\left(\overline{\left|\langle\psi_t|B|\psi_t\rangle - \overline{\langle\psi_t|B|\psi_t\rangle} \right|^2}\right) &\leq 24 \|B\|^2 \|\rho\| D_{E,B} D_{G,B}.\label{eq: exp inf time av}
    \end{align}
\end{prop}

The proof of Proposition~\ref{prop: Var inequalities} can be found in Section~\ref{subsec: proof prop 1}. Before we give the proof of Theorem~\ref{thm: NT GAP}, we first note that for any bounded operator $B$ on $\Hilbert$ such that $\Pi_{e'}B\Pi_e\neq 0$ only for finitely many $e,e'$ and any $\psi_0\in\mathbb{S}(\Hilbert)$ the limit
\begin{align}
    M_{\psi_0 B} = \overline{\langle\psi_t|B|\psi_t\rangle} := \lim_{T\to\infty} \frac{1}{T}\int_0^T \langle\psi_t|B|\psi_t\rangle\, dt
\end{align}
exists and can be computed as follows:
\begin{subequations}
\begin{align}
    M_{\psi_0 B} &= \overline{\left\langle\psi_0\left|e^{iHt}\sum_{e\in\mathcal{E}}\Pi_e B \sum_{e'\in\mathcal{E}}\Pi_{e'} e^{-iHt}\right|\psi_0\right\rangle}\\
    &= \sum_{e,e'\in \mathcal{E}} \overline{e^{i(e-e')t}}\langle\psi_0|\Pi_e B \Pi_{e'}|\psi_0\rangle\\
    &= \sum_{e\in\mathcal{E}}\left\langle\psi_0\left| \Pi_e B \Pi_e\right|\psi_0\right\rangle.\label{eq: Mpsi0B}
\end{align}
\end{subequations}
An application of $\mathbb{E}_\rho\langle\psi|A|\psi\rangle = \tr(\rho A)$ for arbitrary bounded operators $A$, see also Lemma~\ref{lem: var GAP}, immediately shows that
\begin{align}
    \mathbb{E}_\rho M_{\psi_0 B} = \sum_{e\in\mathcal{E}} \tr(\rho\Pi_e B \Pi_e) = M_{\rho B}.
\end{align}
Note that interchanging the sums over $\mathcal{E}$ with the time average and expectation with respect to $\GAP(\rho)$ are unproblematic as due to our assumption on $B$ the sums are effectively sums over only finitely many $e,e'$.
\begin{proof}[Proof of Theorem~\ref{thm: NT GAP}]
We start with proving the bounds which follow from an application of Proposition~\ref{prop: Var inequalities}.
Without loss of generality we assume that all eigenvalues of $\rho$ are positive (otherwise we restrict our considerations to the subspace of $\Hilbert$ spanned by the eigenvectors of $\rho$ corresponding to its positive eigenvalues). 
It follows from Markov's inequality together with Proposition~\ref{prop: Var inequalities} that
\begin{align}
\GAP(\rho)&\left\{\left\langle\left|\langle\psi_t|B|\psi_t\rangle-\overline{\langle\psi_t|B|\psi_t\rangle}\right|^2\right\rangle_T \geq  \frac{48}{\varepsilon}\|B\|^2 \|\rho\| D_{E,B} G_{B}(\kappa)\left(1+\frac{8\log_2 d_{E,B}}{\kappa T}\right)\right\}
\leq \frac{\varepsilon}{2}
\end{align}
and
\begin{align}
    \GAP(\rho)&\left\{\overline{\left|\langle\psi_t|B|\psi_t\rangle-\overline{\langle\psi_t|B|\psi_t\rangle} \right|^2} \geq \frac{48}{\varepsilon} \|B\|^2 \|\rho\| D_{E,B} D_{G,B}\right\}
    \leq \frac{\varepsilon}{2}.
\end{align}
Concerning the finite-time average, this shows that, w.r.t. $\GAP(\rho)$, $(1-\frac{\varepsilon}{2})$-most $\psi_0\in\mathbb{S}(\Hilbert)$ are such that
\begin{align}
    \left\langle\left|\langle\psi_t|B|\psi_t\rangle-\overline{\langle\psi_t|B|\psi_t\rangle}\right|^2\right\rangle_T &\leq \frac{48}{\varepsilon}\|B\|^2 \|\rho\| D_{E,B} G_{B}(\kappa)\left(1+\frac{8\log_2 d_{E,B}}{\kappa T}\right).
\end{align}
Similarly, we have that, w.r.t. $\GAP(\rho)$, $(1-\frac{\varepsilon}{2})$-most $\psi_0\in\mathbb{S}(\Hilbert)$ are such that
\begin{align}
    \overline{\left|\langle\psi_t|B|\psi_t\rangle-\overline{\langle\psi_t|B|\psi_t\rangle} \right|^2} \leq \frac{48}{\varepsilon} \|B\|^2 \|\rho\| D_{E,B} D_{G,B}.
\end{align}
Let $\lambda$ denote the Lebesgue measure on $\mathbb{R}$. An application of Markov's inequality shows that
\begin{align}
    \frac{1}{T}&\lambda\left\{t\in [0,T] : \left|\langle\psi_t|B|\psi_t\rangle-\overline{\langle\psi_t|B|\psi_t\rangle} \right|^2\geq \frac{48}{\varepsilon\delta} \|B\|^2 \|\rho\| D_{E,B} G_B(\kappa)\left(1+\frac{8\log_2 d_{E,B}}{\kappa T}\right)\right\} \leq \delta
\end{align}
for $(1-\frac{\varepsilon}{2})$-most $\psi_0\in\mathbb{S}(\Hilbert)$ w.r.t. $\GAP(\rho)$.
Thus, w.r.t. $\GAP(\rho)$, $(1-\frac{\varepsilon}{2})$-most $\psi_0\in\mathbb{S}(\Hilbert)$ are such that for $(1-\delta)$-most $t\in[0,T]$,
\begin{align}
    &\left|\langle\psi_t|B|\psi_t\rangle-\overline{\langle\psi_t|B|\psi_t\rangle}\right|
    \leq \left(\frac{48}{\varepsilon\delta} \|B\|^2 \|\rho\| D_{E,B} G_B(\kappa)\left(1+\frac{8\log_2 d_{E,B}}{\kappa T}\right)\right)^{1/2}.
\end{align}
In the same way, we obtain
\begin{align}
    \liminf_{T\to\infty} \frac{1}{T}&\lambda\left\{t\in [0,T]: \left|\langle\psi_t|B|\psi_t\rangle-\overline{\langle\psi_t|B|\psi_t\rangle}\right|^2\geq \frac{48}{\varepsilon\delta} \|B\|^2 \|\rho\| D_{E,B} D_{G,B}\right\}\leq \delta
\end{align}
and conclude that, w.r.t. $\GAP(\rho)$, $(1-\frac{\varepsilon}{2})$-most $\psi_0\in\mathbb{S}(\Hilbert)$ are such that for $(1-\delta)$-most $t\in[0,\infty)$,
\begin{align}
    \left|\langle\psi_t|B|\psi_t\rangle-\overline{\langle\psi_t|B|\psi_t\rangle} \right| \leq \left(\frac{48}{\varepsilon\delta}\|B\|^2 \|\rho\| D_{E,B} D_{G,B} \right)^{1/2}.
\end{align}
It follows from Chebyshev's inequality together with Proposition~\ref{prop: Var inequalities} that
\begin{align}
    \GAP(\rho)\left\{\left|\overline{\langle\psi_t|B|\psi_t\rangle} - M_{\rho B} \right|\geq \left(\frac{46\|B\|^2\|\rho\|}{\varepsilon}\right)^{1/2}\right\} \leq \frac{\varepsilon}{2}
\end{align}
and thus, w.r.t. $\GAP(\rho)$, $(1-\frac{\varepsilon}{2})$-most $\psi_0\in\mathbb{S}(\Hilbert)$ are such that
\begin{align}
    \left|\overline{\langle\psi_t|B|\psi_t\rangle} - M_{\rho B}\right| \leq \left(\frac{46\|B\|^2\|\rho\|}{\varepsilon}\right)^{1/2}.
\end{align}
Because of $\sqrt{48}+\sqrt{46}\leq \sqrt{188}$, \eqref{ineq: MR infinite time} and the first bound in \eqref{ineq: MR finite time} now follow from the triangle inequality.

Next we prove the second bound in \eqref{ineq: MR finite time} with the help of Theorem~\ref{thm: Levy Lemma}. Consider the function $f(\psi)=\langle\psi_t|B|\psi_t\rangle$ for fixed $t\geq 0$. By Lemma~5 in \cite{PSW05}, $f$ is Lipschitz continuous with Lipschitz constant bounded by $2\|B\|$. Then it follows from Theorem~\ref{thm: Levy Lemma} that for every $t\in [0,T]$, w.r.t. $\GAP(\rho)$, $(1-\varepsilon)$-most $\psi_0\in\mathbb{S}(\Hilbert)$ are such that
\begin{align}
    \Bigl|\langle\psi_t|B|\psi_t\rangle - \tr\left(e^{iHt} B e^{-iHt}\rho\right) \Bigr| \leq \left(\frac{8\|B\|^2 \|\rho\|\log(12/\varepsilon)}{C}\right)^{1/2}.\label{ineq: LL1}
\end{align}
By Footnote~7 in \cite{GLMTZ10} we can interchange ``most $t$'' and ``most $\psi_0$'' at the expense of making the parameters $\varepsilon$ and $\delta$ worse. More precisely, if a statement is true for $(1-\delta)$-most $t\in [0,T]$ for $(1-\varepsilon)$-most $\psi_0\in\mathbb{S}(\Hilbert)$, then it is also true for $(1-\varepsilon')$-most $\psi_0\in\mathbb{S}(\Hilbert)$ for $(1-\delta')$-most $t\in [0,T]$ where $\varepsilon'\geq \frac{\varepsilon+\delta-\varepsilon\delta}{\delta'}$. Note that this interchange is only possible in the finite-time case as only then both ``most'' refer to measures, see Footnote~\ref{footnote counterex} for a counterexample in the infinite-time case. 

With the help of Footnote~7 in \cite{GLMTZ10} we obtain that, w.r.t. $\GAP(\rho)$, $(1-2\varepsilon/\delta)$-most $\psi_0\in\mathbb{S}(\Hilbert)$ are such that for $(1-\delta/2)$-most $t\in[0,T]$, \eqref{ineq: LL1} holds.
Moreover, Proposition~\ref{prop: Var inequalities} implies that $(1-\delta/2)$-most $t\in[0,T]$ are such that
\begin{align}
    \Bigl|\tr(e^{iHt}Be^{-iHt}\rho)-M_{\rho B} \Bigr| \leq \left(\frac{2\|B\|^2 \|\rho\| D_{E,B} G_B(\kappa)}{\delta}\left(1+\frac{8\log_2 d_{E,B}}{\kappa T}\right)\right)^{1/2}.
\end{align}
Thus $(1-2\varepsilon/\delta)$-most $\psi_0\in\mathbb{S}(\Hilbert)$ are such that for $(1-\delta)$-most $t\in [0,T]$, 
\begin{align}
    \Bigl|\langle\psi_t|B|\psi_t\rangle -M_{\rho B}\Bigr|\leq 5\left(\frac{\|B\|^2 \|\rho\| D_{E,B} G_B(\kappa) \log(12/\varepsilon)}{\delta C}\left(1+\frac{8\log_2 d_{E,B}}{\kappa T}\right)\right)^{1/2},
\end{align}
where we used that $\sqrt{8}+\sqrt{2}\leq 5$.
The second bound in \eqref{ineq: MR finite time} now follows from replacing $\varepsilon$ by $\varepsilon\delta/2$.
\end{proof}

\subsection{Proof of Proposition~\ref{prop: Var inequalities}\label{subsec: proof prop 1}}
The proof of Proposition~\ref{prop: Var inequalities} makes heavy use of an improved upper bound for the GAP-variance of $\langle\psi|A|\psi\rangle$. We state the bound in the following lemma and give its proof in Section~\ref{subsec: proof Lemma 1} 
\begin{lemma}\label{lem: var GAP}
Let $\Hilbert$ be a separable Hilbert space of dimension $\geq 4$.
    Let $\rho$ be a density matrix on $\Hilbert$ with eigenvalues $p_n>0$ such that $p_{\max} = \|\rho\|<1/4$ {and let $\dim\Hilbert\geq 4$}. For $\GAP(\rho)$-distributed $\psi$ and any bounded operator $A:\Hilbert\to\Hilbert$,
    \begin{align}
        &\;\mathbb{E}_\rho\langle\psi|A|\psi\rangle = \tr(A\rho),\label{eq: exp GAP}\\
        &\Var_\rho\langle\psi|A|\psi\rangle \leq  \frac{1}{1-p_{\max}}\left(\tr(A\rho A^* \rho) + \frac{\tr(A\rho^2 A^* \rho)+\tr(A\rho A^*\rho^2)}{1-2p_{\max}}\right.\nonumber\\
    &\left. + \frac{2}{(1-2p_{\max})(1-3p_{\max})} \Biggl[\tr(A\rho^3 A^*\rho) + \tr(A\rho^2 A^* \rho^2)+\tr(A\rho A^*\rho^3)\right.\nonumber\\
    &\left.+ \sum_{m,n}\left(|\tr(A\rho^3 P_m)\tr(A^* \rho P_n)| + |\tr(A\rho^2 P_m)\tr(A^*\rho^2 P_n)| + |\tr(A\rho P_m)\tr(A^* \rho^3 P_n)|\right)\Biggr]\right),\label{eq: var GAP}
    \end{align}
    where $\mathbb{E}_\rho$ and $\Var_\rho$ denote the expectation and variance with respect to $\GAP(\rho)$, $P_n = |n\rangle\langle n|$ and $\{|n\rangle\}$ is an orthonormal eigenbasis of $\rho$.
\end{lemma}

\begin{proof}[Proof of Proposition~\ref{prop: Var inequalities}]
We first assume that $\Hilbert$ is finite-dimensional. The proof starts as the one of Proposition~1 in \cite{TTV23phys} and we find that
\begin{subequations}
\begin{align}
    \left\langle\left|\langle\psi_t|B|\psi_t\rangle - \overline{\langle\psi_t|B|\psi_t\rangle} \right|^2\right\rangle_T
    &=\left\langle\left|\sum_{e\neq e'} e^{i(e-e')t} \langle\psi_0|\Pi_e B \Pi_{e'}|\psi_0\rangle \right|^2\right\rangle_T \\    
    &= \sum_{\substack{e\neq e'\\ e''\neq e'''}} \left\langle e^{i(e-e'-e''+e''')t}\right\rangle_T \langle\psi_0|\Pi_e B \Pi_{e'}|\psi_0\rangle \langle\psi_0|\Pi_{e'''} B^*\Pi_{e''}|\psi_0\rangle\\
    &=:\sum_{\alpha,\beta} v_\alpha^* R_{\alpha\beta} v_\beta,
\end{align}
\end{subequations}
where we defined for $\alpha=(e,e'), \beta = (e'',e''')\in \mathcal{G}:=\{(\bar{e},\bar{e}')\in\mathcal{E}\times \mathcal{E}, \bar{e}\neq \bar{e}'\}$ the vector $v_\alpha=\langle\psi_0|\Pi_{e'}B^* \Pi_e|\psi_0\rangle$ and the Hermitian matrix
\begin{align}
    R_{\alpha\beta} = \left\langle e^{i(G_\alpha-G_\beta)t} \right\rangle_T\label{eq: Ralphabeta}
\end{align}
with $G_\alpha=e-e'$. Written in this way, we immediately see that
\begin{align}
    \left\langle\left|\langle\psi_t|B|\psi_t\rangle - \overline{\langle\psi_t|B|\psi_t\rangle} \right|^2\right\rangle_T \leq \|R\| \sum_{\alpha} |v_\alpha|^2 \leq \|R\| \sum_{e, e'} \left|\langle\psi_0|\Pi_e B \Pi_{e'}|\psi_0\rangle\right|^2.
\end{align}
It was shown rigorously by Short and Farrelly \cite{SF12} that the operator norm of $R$ can be bounded as
\begin{align}
    \|R\| \leq G(\kappa) \left(1+\frac{8\log_2 d_E}{\kappa T}\right).
\end{align}
We further compute
\begin{align}
    \mathbb{E}_\rho\left(\sum_{e,e'} |\langle\psi_0|\Pi_e B \Pi_{e'}|\psi_0\rangle|^2\right) 
    &=\sum_{e,e'} \Var_\rho\langle\psi_0|\Pi_e B \Pi_{e'}|\psi_0\rangle + \left|\tr(\rho \Pi_e B \Pi_{e'}) \right|^2,
\end{align}
where we applied Lemma~\ref{lem: var GAP}. It follows again from Lemma~\ref{lem: var GAP} that
\begin{align}
    &\sum_{e,e'}\Var_\rho\langle\psi_0|\Pi_e B \Pi_{e'}|\psi_0\rangle\nonumber\\
    &\leq \frac{1}{1-p_{\max}}\sum_{e,e'}\Bigl(\tr(\Pi_e B \Pi_{e'}\rho \Pi_{e'}B^* \Pi_e \rho) + \frac{\tr(\Pi_e B \Pi_{e'}\rho^2 \Pi_{e'}B^* \Pi_e \rho) + \tr(\Pi_{e} B \Pi_{e'} \rho \Pi_{e'}B^*\Pi_e \rho^2)}{1-2p_{\max}}\nonumber\\
    &\quad + \frac{2}{(1-2p_{\max})(1-3p_{\max})}\left[\tr(\Pi_e B \Pi_{e'}\rho^3 \Pi_{e'} B^* \Pi_e \rho) + \tr(\Pi_e B \Pi_{e'}\rho^2 \Pi_{e'} B^* \Pi_e \rho^2)\right.\nonumber\\
    &\quad \left. + \tr(\Pi_e B \Pi_{e'}\rho \Pi_{e'} B^* \Pi_e \rho^3)+ \sum_{m,n}\left(|\tr(\Pi_e B \Pi_{e'}\rho^3 P_m) \tr(\Pi_{e'}B^*\Pi_e\rho P_n)|\right.\right.\nonumber\\
    &\quad \left.\left.+ |\tr(\Pi_e B \Pi_{e'}\rho^2 P_m) \tr(\Pi_{e'}B^*\Pi_e\rho^2 P_n)| + |\tr(\Pi_e B \Pi_{e'}\rho P_m) \tr(\Pi_{e'} B^* \Pi_e \rho^3 P_n)|\right) \right]\Bigr).\label{ineq: sum Var}
\end{align}
The main tools for bounding the traces are the Cauchy-Schwarz inequality for the trace which states that $|\tr(A^*B)|\leq\sqrt{\tr(A^*A)\tr(B^*B)}$ for any operators $A$ and $B$ and the inequality $|\tr(AB)|\leq \|A\| \tr(|B|)$; for the latter see, e.g., Theorem~3.7.6 in \cite{simon}. We compute
\begin{subequations}
\begin{align}
    \sum_{e,e'} \tr(\Pi_e B \Pi_{e'}\rho \Pi_{e'} B^* \Pi_e \rho) &= \sum_e \tr\left(B \left(\sum_{e'}\Pi_{e'}\rho \Pi_{e'}\right) B^* \Pi_e \rho \Pi_e\right)\\
    &\leq \sum_e \|B\|^2 \left\|\sum_{e'} \Pi_{e'}\rho\Pi_{e'}\right\| \tr(\Pi_e \rho \Pi_e)\\
    &\leq \|B\|^2 \|\rho\|,
\end{align}
\end{subequations}
where we used that $\|\sum_{e'}\Pi_{e'}\rho\Pi_{e'}\|\leq \|\rho\|$. All other sums over traces in \eqref{ineq: sum Var} which involve no $P_n$ can be estimated in a similar way.

Next we turn to the products of traces; we find that
\begin{subequations}
\begin{align}
    \sum_{e,e'} \sum_{m,n} &|\tr(\Pi_e B \Pi_{e'}\rho^3 P_m)\tr(\Pi_{e'}B^*\Pi_e\rho P_n)|\nonumber\\
    &= \sum_{m,n}\sum_{e,e'} |\tr(\Pi_e B \Pi_{e'} \Pi_{e'} P_m \Pi_e)\tr(\Pi_{e'}B^*\Pi_e \Pi_e P_n \Pi_{e'})| p_m^3 p_n\\
    &\leq \sum_{m,n} p_m^3 p_n\sum_{e,e'} \left(\tr(\Pi_e B \Pi_{e'} B^*) \tr(\Pi_e P_m \Pi_{e'} P_m)\tr(\Pi_{e'}B^*\Pi_e B)\tr(\Pi_{e'}P_n \Pi_e P_n)\right)^{1/2}\\
    &= \sum_{m,n} p_m^3 p_n \sum_{e,e'} \underbrace{\tr(\Pi_e B \Pi_{e'} B^*)}_{\leq \|B\|^2 D_E}\left(\langle m|\Pi_e|m\rangle\langle m|\Pi_{e'}|m\rangle\langle n|\Pi_e|n\rangle\langle n|\Pi_{e'}|n\rangle\right)^{1/2}\\
    &\leq \|B\|^2 D_E \sum_{m,n} p_m^3 p_n \left(\sum_{e}\left(\langle m|\Pi_e|m\rangle \langle n|\Pi_e|n\rangle\right)^{1/2}\right)^2\\
    &\leq \|B\|^2 D_E \sum_{m,n} p_m^3 p_n \left(\sum_e \langle m|\Pi_e|m\rangle\right) \left(\sum_e \langle n|\Pi_e|n\rangle\right)\\
    &\leq \|B\|^2 D_E \|\rho\|^2
\end{align}
\end{subequations}
and similarly for the other products of traces.

Finally we estimate
\begin{align}
    \sum_{e,e'} |\tr(\rho \Pi_e B \Pi_{e'})|^2 
    &\leq \sum_{e,e'} \tr(\Pi_e B \Pi_{e'} B^*) \tr(\Pi_e \rho \Pi_{e'}\rho) \leq \|B\|^2 D_E \|\rho\|.
\end{align}
Putting everything together we arrive at
\begin{subequations}
\begin{align}
    \mathbb{E}_\rho&\left(\left\langle\left|\langle\psi_t|B|\psi_t\rangle - \overline{\langle\psi_t|B|\psi_t\rangle} \right|^2\right\rangle_T\right)\nonumber\\  &\leq \|B\|^2 \|\rho\| D_E G(\kappa) \left(1+\frac{8\log_2 d_E}{\kappa T}\right)\left(1+\frac{1}{1-\|\rho\|}\left[1+\frac{2\|\rho\|}{1-2\|\rho\|} + \frac{6(\|\rho\|+\|\rho\|^2)}{(1-2\|\rho\|)(1-3\|\rho\|)}\right]\right)\label{ineq: Erho}\\
    &\leq 24 \|B\|^2 \|\rho\| D_E G(\kappa)\left(1+\frac{8\log_2 d_E}{\kappa T}\right),\label{ineq: Erho time var}
\end{align}
\end{subequations}
where we used the assumption that $\|\rho\|< 1/4$ to bound the bracket $(1+\frac{1}{1-\|\rho\|}[...])$ in \eqref{ineq: Erho} by\footnote{Note that the bracket is a monotone increasing function in $\|\rho\|$ and therefore bounded by its value for $\|\rho\|=1/4$ which is $71/3<24$.} 24. The bound in \eqref{ineq: prop1 second} can be obtained by estimates similar to the ones in the beginning of this proof. For the statement regarding the infinite time average, we want to take the limit $T\to\infty$ in \eqref{ineq: Erho time var}. To this end, we first choose $\kappa$ so small that $G(\kappa)=D_G$. Then we take the limit $T\to\infty$ and use dominated convergence to interchange the limit and $\mathbb{E}_\rho$. This gives \eqref{eq: exp inf time av}.

For the variance of $\overline{\langle\psi_t|B|\psi_t\rangle}$ we see from \eqref{eq: Mpsi0B} that we can simply apply Lemma~\ref{lem: var GAP} with $A=\sum_{e}\Pi_e B \Pi_e$ and estimate the occurring traces. We compute
\begin{align}
    \tr(A\rho A^* \rho) \leq \| A \rho A^*\| \tr\rho \leq \|B\|^2 \|\rho\|,
\end{align}
where we used that $\|A\| \leq \|B\|$. The other traces without $P_n$ can be estimated analogously. 

Moreover, we get
\begin{align}
    \sum_{m,n} \left|\tr(A\rho^3 P_m) \tr(A^*\rho P_n) \right| \leq \sum_{m,n} p_m^3 p_n \|B\|^2 \tr(P_m) \tr(P_n) \leq \|B\|^2 \|\rho\|^2
\end{align}
and similarly for the other remaining terms.

Thus we finally arrive at
\begin{subequations}
\begin{align}
    \Var_\rho\overline{\langle\psi_t|B|\psi_t\rangle} &\leq \frac{\|B\|^2 \|\rho\|}{1-\|\rho\|}\left(1+\frac{2\|\rho\|}{1-2\|\rho\|} + \frac{6(\|\rho\|+\|\rho\|^2)}{(1-2\|\rho\|)(1-3\|\rho\|)}\right)\leq 23 \|B\|^2 \|\rho\|.
\end{align}
\end{subequations}
Now suppose that $\Hilbert$ is infinite-dimensional. Due to our assumption on $B$, all sums over $e\in\mathcal{E}$ are effectively sums over $e\in\mathcal{E}_B$ and therefore finite. In the finite-dimensional setting, the matrix $R$ defined by \eqref{eq: Ralphabeta} is a $d_E(d_E-1)\times d_E(d_E-1)$ matrix and this is where the $d_E$ in the upper bound for $\|R\|$ comes from, see the proof in \cite{SF12}. In the infinite-dimensional setting, the matrix $R$ can be viewed as a $d_{E,B}(d_{E,B}-1)\times d_{E,B}(d_{E,B}-1)$ matrix and therefore, in the bound for $\|R\|$, $d_E$ has to be replaced by $d_{E,B}$; by the same argument, $d_E$ can also be replaced by $d_{E,B}$ in the finite-dimensional setting. Moreover, we can also change $G(\kappa)$ to $G_B(\kappa)$, $D_E$ to $D_{E,B}$ and $D_G$ to $D_{G,B}$ as the only ``contributing'' eigenvalues are the ones in $\mathcal{E}_{B}$. Besides these changes all other steps of the proof remain valid and this proves the bounds also in the infinite-dimensional case.
\end{proof}

\subsection{Proof of Lemma~\ref{lem: var GAP}\label{subsec: proof Lemma 1}}
\begin{proof}
The formula for the expectation has already been proven in \cite[Proposition 1]{TTV24}. For the proof of \eqref{eq: var GAP} we adapt the proof for the upper bound on the variance from \cite{TTV24} which in turn follows closely the proof of Reimann~\cite{Reimann08} who only considered the case that $A$ is self-adjoint and $\Hilbert$ is finite-dimensional. We remark that Reimann's arguments are mostly rigorous and whenever we make use of a result from his paper, we provide, if necessary, the details that are needed to make the arguments fully rigorous in a footnote.

We first assume that $D:=\dim\Hilbert<\infty$.
Recall that the variance of a complex-valued random variable $Z$ is defined as
\begin{align}
    \Var_\rho(Z) = \mathbb{E}_\rho(|Z-\mathbb{E}_\rho Z|^2) = \mathbb{E}_\rho|Z|^2 -|\mathbb{E}_\rho Z|^2.
\end{align}
Because of the invariance of $\Var_\rho$ under the addition of constants, i.e., $\Var_\rho(Z+a) = \Var_\rho(Z)$ for all constants $a\in\mathbb{C}$, we can assume without loss of generality that $\mathbb{E}_\rho\langle\psi|A|\psi\rangle = \tr(A\rho)=0$. As a consequence, we only have to compute $\mathbb{E}_\rho(\langle\psi|A|\psi\rangle \langle\psi|A|\psi\rangle^*) = \mathbb{E}_\rho(\langle\psi|A|\psi\rangle\langle\psi|A^*|\psi\rangle)$.

Let $\{|n\rangle: n=1,\dots,D\}$ be an orthonormal basis of $\Hilbert$ consisting of eigenvectors of $\rho$ and let $\psi\in\mathbb{S}(\Hilbert)$. Then,
\begin{subequations}
\begin{align}
    \langle\psi|A|\psi\rangle \langle\psi|A^*|\psi\rangle &= \sum_{m,n,m',n'} \langle\psi|m\rangle\langle m|A|n\rangle \langle n|\psi\rangle \langle\psi|n'\rangle\langle n'|A^*|m'\rangle\langle m'|\psi\rangle\\
    &= \sum_{m,n,m',n'} c_m^* c_n c_{n'}^* c_{m'} A_{mn} A_{m'n'}^*,
\end{align}
\end{subequations}
where $c_m=\langle m|\psi\rangle$ and $A_{mn} = \langle m|A|n\rangle$.

Reimann~\cite{Reimann08} computed\footnote{The only not fully rigorous step in his computation is when he interchanges integration and partial derivatives in \cite[(43)]{Reimann08}. This step can easily be made rigorous by already assuming here that the dimension satisfies $D\geq 4$, an assumption that is made later in \cite[(54)]{Reimann08} anyway. Then one can check that all conditions of the Leibniz integral rule are satisfied which justifies the interchange of integration and differentiation.} the fourth moments $\mathbb{E}_\rho (c_m^* c_n c_{n'}^* c_{m'})$ and obtained that they vanish except in the two cases (i) $m=n$ and $m'=n'$, (ii) $m=m'$ and $n=n'$, and that
\begin{align}
    \mathbb{E}_\rho \left(|c_m|^2 |c_n|^2\right) = p_m p_n (1+\delta_{mn}) K_{mn},
\end{align}
where 
\begin{align}
    K_{mn} = \int_0^\infty (1+xp_m)^{-1} (1+xp_n)^{-1} \prod_{l=1}^D (1+xp_l)^{-1}\, dx.
\end{align}
We compute
\begin{subequations}
\begin{align}
    \mathbb{E}_\rho\left(\langle\psi|A|\psi\rangle\langle\psi|A^*|\psi\rangle\right) &= \sum_{m,n} p_n p_m (1+\delta_{mn}) K_{mn} A_{mm} A_{nn}^* + \sum_{m,n} p_n p_m (1+\delta_{mn}) K_{mn} |A_{mn}|^2 \nonumber\\
    &\quad- 2 \sum_n p_n^2 K_{nn} |A_{nn}|^2\\
    &= \sum_{m,n}\left[A_{mm} A_{nn}^* + |A_{mn}|^2\right] p_n p_m K_{mn}\label{eq: Erho psiApsi}
\end{align}
\end{subequations}
and as in \cite{Reimann08} we write $K_{mn}$ as
\begin{align}
    K_{mn} = K^{(0)} - (p_m+p_n)K^{(1)} + 2 \kappa_{mn}\left(p_m^2+p_m p_n + p_n^2\right)K^{(2)},\label{eq: Kmn expansion}
\end{align}
where 
\begin{align}
    K^{(k)} = \frac{1}{k!} \int_0^\infty x^k \prod_{l=1}^D (1+xp_l)^{-1} dx, \quad k=0,1,2,
\end{align}
and $\kappa_{mn}\in [0,1]$. This follows rigorously from a Taylor expansion of $g_{mn}(x):=(1+xp_m)^{-1}(1+xp_n)^{-1}$ around 0 up to second order.

With \eqref{eq: Kmn expansion} we obtain for the first term in \eqref{eq: Erho psiApsi} that 
\begin{align}
   \sum_{m,n} A_{mm} &A_{nn}^* p_n p_m \left(K^{(0)} - (p_m+p_n) K^{(1)} + 2\kappa_{mn}(p_m^2+p_m p_n + p_n^2) K^{(2)}\right)\nonumber\\
   &= 2 K^{(2)} \sum_{m,n} A_{mm} A_{nn}^* \kappa_{mn}\left(p_m^3 p_n + p_m^2 p_n^2 + p_m p_n^3\right),\label{eq: var gap first}
\end{align}
where we used that $\sum_m A_{mm} p_m = \tr(A\rho)=0$
and similarly $\sum_n A^*_{nn} p_n=0$. Therefore an upper bound for the absolute value of the first term in \eqref{eq: Erho psiApsi} is given by
\begin{align}
    &2 K^{(2)}\sum_{m,n}\left(|\tr(A\rho^3 P_m)\tr(A^* \rho P_n)| + |\tr(A\rho^2 P_m)\tr(A^*\rho^2 P_n)| + |\tr(A\rho P_m)\tr(A^* \rho^3 P_n)|\right),
\end{align}
where $P_n = |n\rangle\langle n|$. Here we used e.g. for the first term that $A_{mm} A^*_{nn} p_m^3 p_n = \tr(A\rho^3 P_m) \tr(A^*\rho P_n)$.

With the help of similar computations the second term in \eqref{eq: Erho psiApsi} can be written as
\begin{subequations}
\begin{align}
    &K^{(0)} \tr(A\rho A^*\rho) - K^{(1)}\left(\tr(A\rho^2 A^*\rho)+\tr(A\rho A^*\rho^2)\right)\nonumber\\
    &\quad + 2 K^{(2)} \sum_{m,n} |A_{mn}|^2 \kappa_{mn}(p_n p_m^3 + p_m^2 p_n^2 + p_n^3 p_m)\label{eq: var gap second}\\
&\leq K^{(0)} \tr(A\rho A^* \rho) + K^{(1)} (\tr(A\rho^2 A^*\rho) + \tr(A\rho A^*\rho^2))\nonumber\\
    &\quad + 2 K^{(2)} \left(\tr(A\rho^3 A^*\rho) + \tr(A\rho^2 A^* \rho^2)+\tr(A\rho A^*\rho^3)\right).
\end{align}
\end{subequations}
Thus we arrive at
\begin{align}
    &\Var_\rho \langle\psi|A|\psi\rangle\nonumber\\
    &\quad\leq K^{(0)} \tr(A\rho A^* \rho) + K^{(1)} \left(\tr(A\rho^2 A^*\rho)+\tr(A\rho A^*\rho^2) \right)\nonumber\\
    &\quad+ 2 K^{(2)}\left(\tr(A\rho^3 A^*\rho) + \tr(A\rho^2 A^* \rho^2)+\tr(A\rho A^*\rho^3)\right.\nonumber\\
    &\quad\left. + \sum_{m,n}\left(|\tr(A\rho^3 P_m)\tr(A^* \rho P_n)| + |\tr(A\rho^2 P_m)\tr(A^*\rho^2 P_n)| + |\tr(A\rho P_m)\tr(A^* \rho^3 P_n)|\right)\right).
\end{align}
It follows from the estimates\footnote{More precisely, Reimann shows that for any $x\geq 0$ the function $G(x) := \prod_{l=1}^D (1+xp_k)^{-1}$ can be bounded as
\begin{align}
    G(x) \leq (1+xp_{\max})^{-1/p_{\max}},\label{ineq: G bound}
\end{align}
see \cite[(70)]{Reimann08}.
Then \eqref{ineq: K bound} follows after performing $k$ partial integrations. For the proof of \eqref{ineq: G bound} Reimann first heuristically argues that 
\begin{align}
    G(x) \leq (1+xp_{\max})^{-N_{\max}} (1+xp_0)^{-1},\label{ineq: G intermediate}
\end{align}
where $N_{\max} = \lfloor \frac{1}{p_{\max}}\rfloor$ and $p_0:=1-N_{\max}p_{\max}$ and then rigorously derives \eqref{ineq: G bound} from \eqref{ineq: G intermediate}. The bound in \eqref{ineq: G intermediate} can be seen rigorously as follows: For
$p=(p_1,\dots,p_D) \in \mathcal{C}:=\left\{y\in\mathbb{R}^D: 0\leq y_n\leq p_{\max} \, \forall n \mbox{ and } \sum_l y_l=1\right\}$ and fixed $x\geq 0$ we define $Q(p):= \prod_{l=1}^D (1+xp_l)^{-1}$. Maximizing $Q(p)$ is equivalent to minimizing the function $S(p) := \sum_{l=1}^D \ln(1+xp_l)$. Since $S$ is continuous and concave and $\mathcal{C}$ is convex and compact, it follows from Bauer's minimum principle that $S$ attains its minimum at an extreme point of $\mathcal{C}$. The extreme points of $\mathcal{C}$ are the permutations of $(p_{\max},\dots,p_{\max}, p_0, 0,\dots)$ and this immediately implies~\eqref{ineq: G intermediate}.
} in \cite{Reimann08} that
\begin{align}
    K^{(k)} \leq \prod_{j=1}^{k+1} \frac{1}{1-jp_{\max}}, \quad k=0,1,2.\label{ineq: K bound}
\end{align}
With the help of this inequality we finally obtain \eqref{eq: var GAP} in the finite-dimensional case.

Now suppose that $\dim\Hilbert = \infty$. We proceed similarly as in the proof of Proposition~1 in \cite{TTV24}. Let $\{|k\rangle : k\in\mathbb{N}\}$ be an orthonormal eigenbasis of $\rho$ with corresponding positive eigenvalues $p_{\max} = \|\rho\| = p_1\geq p_2\geq \dots~$. We approximate $\rho$ in trace norm by finite-rank density matrices $\rho_k$, $k\in\mathbb{N}$, defined by
\begin{align}
    \rho_k := \sum_{m=1}^{k-1} p_m|m\rangle\langle m| + \left(\sum_{m=k}^{\infty}p_m\right) |k\rangle\langle k|.
\end{align}
Because of $\|\rho_k-\rho\|_{\tr}\to 0$ as $k\to\infty$, it follows from Theorem~3 in \cite{Tum20} that $\GAP(\rho_k)\Rightarrow \GAP(\rho)$, i.e., the measures $\GAP(\rho_k)$ converge weakly to the measure $\GAP(\rho)$. Note that for $k$ large enough, $\sum_{m\geq k} p_m < p_1 = p_{\max}$ and therefore $p_{\max,k} = p_{\max}$ where $p_{\max,k}$ denotes the largest eigenvalue of $\rho_k$. We define the functions $f_k,f:\mathbb{S}(\Hilbert)\to\mathbb{R}$ by
\begin{align}
    f_k(\psi) &= \left|\langle\psi|A|\psi\rangle - \tr(A\rho_k)\right|^2,\\
    f(\psi) &= \left|\langle\psi|A|\psi\rangle- \tr(A\rho)\right|^2.
\end{align}
It follows from
\begin{align}
    \left|\tr(A\rho_k)-\tr(A\rho) \right| \leq \|A\| \|\rho_k-\rho\|_{\tr} \to 0 \label{ineq: conv tr A rhok}
\end{align}
that $f_k \to f$ uniformly in $\psi$ and this implies $\GAP(\rho_k)(f_k) - \GAP(\rho_k)(f)\to 0$ where we introduced the notation
\begin{align}
    \GAP(\rho)(f) := \int_{\mathbb{S}(\Hilbert)} f(\psi)\; d\GAP(\rho)(\psi).
\end{align}
As the measures $\GAP(\rho_k)$ converge weakly to $\GAP(\rho)$ and $f$ is continuous, we have that $\GAP(\rho_k)(f)\to \GAP(\rho)(f)$ and therefore $\Var_{\rho_k}\langle\psi|A|\psi\rangle = \GAP(\rho_k)(f_k) \to \GAP(\rho)(f) = \Var_\rho\langle\psi|A|\psi\rangle$. By restricting to the subspace spanned by $\{|n\rangle: n=1,\dots,k\}$ we obtain from the finite-dimensional case that \eqref{eq: var GAP} holds with $\rho$ and $p_{\max}$ replaced by $\rho_k$ and $p_{\max,k}$ respectively,
where we chose $k$ large enough such that $p_{\max,k}<1/4$. We have already discussed that $p_{\max,k}=p_{\max}$ for $k$ large enough. Moreover, using $[\rho_k,\rho]=0$, we find 
\begin{align}
    \|\rho_k^2-\rho^2\|_{\tr} &= \left\|(\rho_k+\rho)(\rho_k-\rho)\right\|_{\tr} \leq \|\rho_k+\rho\| \|\rho_k-\rho\|_{\tr} \leq 2 \|\rho_k-\rho\|_{\tr} \to 0,
\end{align}
and similarly $\|\rho_k^3-\rho^3\|_{\tr}\to 0$. Therefore it follows from estimates as in \eqref{ineq: conv tr A rhok} that all traces in \eqref{eq: var GAP} for $\rho_k$ converge to the corresponding traces with $\rho$ instead of $\rho_k$. Altogether we thus see by taking the limit $k\to\infty$ in \eqref{eq: var GAP} for $\rho_k$ that the upper bound for the variance remains true also in the infinite-dimensional case.
\end{proof}

\section{Summary and Conclusions\label{sec: conclusion}}
Our result concerns the long-time behavior of $\GAP(\rho)$-typical pure states $\psi_0$ from the sphere of a separable Hilbert space where $\rho$ is a density matrix on $\Hilbert$ and $\psi_0$ evolves unitarily according to $\psi_t=\exp(-iHt)\psi_0$. We have seen that for any bounded operator $B$ on $\Hilbert$ that satisfies $|\mathcal{E}_B|<\infty$, for $\GAP(\rho)$-most $\psi_0\in\mathbb{S}(\Hilbert)$, the curve $t\mapsto \langle\psi_t|B|\psi_t\rangle$ is nearly constant in the long run $t\to\infty$ provided that $\|\rho\|$ is small, $\|B\|$ is not too large compared to $|M_{\rho B}|$ and no eigenvalue or eigenvalue gap of the Hamiltonian $H$, which is assumed to have pure point spectrum, is too highly degenerate. In particular, we have argued that our result can be applied to the case that $B=P_\nu$ for some macro state $\nu$. We have provided explicit error bounds that reveal the dependence of the error on these quantities. Our result shows that the concept of normal typicality can be generalized from the uniform distribution on the sphere of the Hilbert space to a much broader class of distributions which can also be defined on infinite-dimensional Hilbert spaces and which for certain density matrices arise naturally as the distribution of wave functions in thermal and chemical equilibrium.

\bigskip

\noindent \textbf{Acknowledgments.} It is a pleasure to thank Stefan Teufel and Roderich Tumulka for their helpful discussions. Financial support from the German Academic Scholarship Foundation is gratefully acknowledged. Moreover, this work has been partially supported by the ERC Starting Grant ``FermiMath", grant agreement nr. 101040991, funded by the European Union. Views and opinions expressed are however those of the author only and do not necessarily reflect those of the European Union or the European Research Council Executive Agency. Neither the European Union nor the granting authority can be held responsible for them.
\\
\\
\noindent\textbf{Data Availability Statement.}
Data sharing is not applicable to this article as no datasets were generated or analyzed.
\\
\\
\noindent\textbf{ Conflict of Interest Statement.}
The authors have no conflicts to disclose.

\bibliographystyle{plainurl}
\bibliography{LiteratureNTGAP.bib}

\end{document}